\documentclass[a4paper]{article}

\usepackage[english]{babel}
\usepackage{verbatim}
\usepackage{chicago}
\usepackage[a4paper,top=3cm,bottom=2cm,left=3cm,right=3cm,marginparwidth=1.75cm]{geometry}
\usepackage[toc,page]{appendix}
\usepackage[utf8]{inputenc} 
\usepackage[T1]{fontenc}    
\usepackage[urlcolor=true]{hyperref}       
\usepackage{url}            
\usepackage{booktabs}       
\usepackage{amsfonts}       
\usepackage{nicefrac}       
\usepackage{microtype}      
\usepackage{xcolor}
\usepackage{comment}

\usepackage{graphicx}
\usepackage{amssymb}
\usepackage{amsmath}
\usepackage{enumitem,kantlipsum}
\usepackage[english]{babel}
\usepackage{amsthm}
\usepackage{bbm}
\allowdisplaybreaks
\newtheorem{theorem}{Theorem}

\newtheorem{lemma}{Lemma}
\newtheorem*{remark}{Remark}
\AtEndDocument{\refstepcounter{theorem}\label{finalthm}}
\AtEndDocument{\refstepcounter{lemma}\label{finallmm}}

\newtheorem{innercustomgeneric}{\customgenericname}
\providecommand{\customgenericname}{}
\newcommand{\newcustomtheorem}[2]{%
  \newenvironment{#1}[1]
  {%
   \renewcommand\customgenericname{#2}%
   \renewcommand\theinnercustomgeneric{##1}%
   \innercustomgeneric
  }
  {\endinnercustomgeneric}
}

\newcustomtheorem{customthm}{Theorem}
\newcustomtheorem{customlm}{Lemma}

\title{\bf A New Family of Tractable Ising Models}

\author{%
  Valerii Likhosherstov$^{1}$, Yury Maximov$^{(2,1)}$ and Michael Chertkov$^{(3,2,1)}$\\
$^{1}$ Skolkovo Institute of Science and Technology, Moscow, Russia\\
$^{2}$ Theoretical Division and Center for Nonlinear Studies,\\ Los Alamos National Laboratory, Los Alamos, NM, USA\\
$^{3}$ Graduate Program in Applied Mathematics,\\ University of Arizona, Tucson, AZ, USA
}

\begin{document}

\maketitle

\begin{abstract}
We present a new family of zero-field Ising models over $N$ binary variables/spins obtained by consecutive ``gluing'' of planar and $O(1)$-sized components along with subsets of at most three vertices into a tree. The polynomial time algorithm of the dynamic programming type for solving exact inference (partition function computation) and sampling consists of a sequential application of an efficient (for planar) or brute-force (for $O(1)$-sized) inference and sampling to the components as a black box. To illustrate utility of the new family of tractable graphical models, we first build an $O(N^\frac32)$ algorithm for inference and sampling of the $K_5$-minor-free zero-field Ising models---an extension of the planar zero-field Ising models---which is neither genus- nor treewidth-bounded. Second, we demonstrate empirically an improvement in the approximation quality of the NP-hard problem of the square-grid Ising model (with non-zero field) inference.
\end{abstract}

\section{Introduction}

Let $G=(V, E)$ be an undirected graph with a set of vertices $V(G)$ and a set of normal edges $E(G)$ (no loops or multiple edges). We discuss \textit{Ising models} which associate the following probability to each random $N\triangleq |V(G)|$-dimensional binary variable/spin configuration $X \in  \{\pm 1 \}^{N}$: 
\begin{gather}
    \mathbb{P}(X)  \triangleq \frac{\mathcal{W}(X)}{Z}, \label{eq:zfim}
\end{gather}
where 
\begin{gather}
    \mathcal{W}(X) \triangleq  \exp\biggl( \sum_{v \in V(G)} \mu_v x_v + \sum_{ e = \{ v, w \} \in E(G)} J_e x_v x_w \biggr) 
   \quad \text{and}\quad Z \triangleq \sum_{X \in \{ \pm 1 \}^N} \mathcal{W}(X). \label{eq:Z} 
\end{gather}
Here, $\mu = (\mu_v, v \in V(G))$ is a vector of \textit{(magnetic) fields}, $J = (J_e, e \in E(G))$ is a vector of the \textit{(pairwise) spin interactions}, and the normalization constant $Z$, which is defined as a sum over $2^N$ spin configurations, is referred to as the \textit{partition function}. Given the model specification $\mathcal{I} = \langle G, \mu, J \rangle$, we address the tasks of finding the exact value of $Z$ (inference) and drawing exact samples with the probability (\ref{eq:zfim}).

\textbf{Related work.} 
It has been known since the seminal contributions of Fisher \cite{fisher} and Kasteleyn \cite{kasteleyn} that computation of the partition function in the zero-field ($\mu = 0$) Ising model over a planar graph and sampling from the respective probability distribution are both tractable, that is, these are tasks of complexity polynomial in $N$.
As shown by Barahona \cite{barahona}, even when $G$ is planar or when $\mu = 0$ (\textit{zero field}),
the positive results are hard to generalize---both addition of the non-zero (magnetic) field  and the extension beyond planar graphs make
the computation of the partition function NP-hard. These results are also consistent with the statement from Jerrum and Sinclair \cite{jerrum-sinclair} that computation of the partition function of the zero-field Ising model is a \#P-complete problem, even in the ferromagnetic case when all components of $J$ are positive. Therefore, describing $\langle G, \mu, J \rangle$ families for which computations of the partition function and sampling are tractable remains an open question.

The simplest tractable (i.e., inference and sampling are polynomial in $N$) example is one when $G$ is a tree, and the corresponding inference algorithm, known as \textit{dynamic programming} and/or \textit{belief propagation}, has a long history in physics \cite{bethe,peierls}, optimal control \cite{bellman}, information theory \cite{gallager}, and artificial intelligence \cite{pearl}. Extension to the case when $G$ is a tree of $(t + 1)$-sized cliques ``glued'' together, or more formally when $G$ is of a \textit{treewidth}~$t$, is known as the \textit{junction tree algorithm} \cite{jensen}, which has complexity of counting and sampling that grow exponentially with $t$.

Another insight 
originates from the foundational statistical physics literature of the last century related to a zero-field version of (\ref{eq:zfim}) when $G$ is planar. Onsager \cite{onsager} found a closed-form solution of (\ref{eq:zfim}) in the case of a homogeneous Ising model over an infinite two-dimensional square grid. Kac and Ward \cite{kac-ward} reduced the inference of (\ref{eq:zfim}) over a finite square lattice to computing a determinant. Kasteleyn \cite{kasteleyn} generalized this result to an arbitrary planar graph. Kasteleyn's approach consists of expanding each vertex of $G$ into a gadget and reducing the Ising model inference to the problem of counting perfect matchings over the expanded graph. Kasteleyn's construction was simplified by Fisher in \cite{fisher}. The tightest running time estimate for Kasteleyn's method gives $O(N^\frac32)$. Kasteleyn conjectured, which was later proven in~\cite{gallucio}, that the approach extends to the case of the zero-field Ising model over graphs embedded in a surface of \textit{genus} $g$ with a multiplicative $O(4^g)$ penalty. A slightly different reduction to perfect matching counting \cite{barahona,bieche,schraudolph-kamenetsky} also allows one to implement $O(N^{\frac{3}{2}})$ sampling of planar zero-field Ising models using Wilson's algorithm \cite{k33,wilson}. A $K_{33}$ (Figure \ref{fig:k5}(a)) minor-free extension of planar zero-field inference and sampling was constructed in \cite{k33}.

An upper-bound approximation to a general class of inference problems can be built by utilizing the family of tractable spanning Ising submodels---either trees \cite{wainwright} or planar topologies \cite{globerson}.



\textbf{Contribution.} In this manuscript, we first describe a new family of zero-field Ising models on graphs that are more general than planar. Given a tree decomposition of such graphs into planar and ``small'' ($O(1)$-sized) components ``glued'' together along sets of at most three vertices, inference and sampling over the new family of models is of polynomial time. We further show that all the $K_5$-minor-free graphs are included in this family and, moreover, their aforementioned tree decomposition can be constructed with $O(N)$ efforts. (See Figure \ref{fig:k5}(a) for an illustration.)  This allows us to prove an $O(N^\frac32)$ upper bound on run time complexity of inference and sampling of the $K_5$-free zero-field Ising models. Planar graphs are included in the set of $K_5$-free graphs, which are neither genus- nor treewidth-bounded in general.



Second, we show how the newly introduced tractable family of zero-field Ising models allows extension of the approach of \cite{globerson} to an upper-bound log-partition function of arbitrary Ising models. Instead of using planar spanning subgraphs as in \cite{globerson}, we utilize more general (nonplanar) basic tractable elements. Using the methodology of \cite{globerson}, we illustrate the approach through experiments with a nonzero-field Ising model on a square grid for which inference is NP-hard \cite{barahona}.

\textbf{Relation to other algorithms.} The result presented in this manuscript is similar to the approach used to count perfect matchings in $K_5$-free graphs \cite{curticapean,straub}. However, we do not use a transition to perfect matching counting as it is typically done in studies of zero-field Ising models over planar graphs~\cite{fisher,kasteleyn,thomas-middleton1}. 
Presumably, a direct transition to perfect matching counting can be done via a construction of an expanded graph in the fashion of \cite{fisher,kasteleyn}. However, this results in a size increase and, what's more important, there is no direct correspondence between spin configurations and perfect matchings, so sampling is not supported.


Our approach can also be viewed as extending results reported in \cite{k33} on the inference and sampling in the $K_{33}$-free zero-field Ising models. In \cite{k33}, $K_{33}$-free graphs are decomposed into planar and $K_5$ components along pairs of vertices, and the whole construction relies on the underlying planar perfect matching model. In this manuscript, we reformulate the $K_{33}$-free construction of \cite{k33} directly in terms of the Ising model bypassing mapping to perfect matchings. 
Moreover, an extension of the decomposition to gluing over triplets of vertices generalizes the construction, in particular, yielding novel results for efficient inference and learning for the zero-field Ising models over $K_5$-free graphs.




\textbf{Structure.} Section \ref{sec:main} formally introduces the concept of the so-called $c$-nice decomposition of graphs and formulates and proves tractability of $c$-nice decomposable zero-field Ising models.
Section \ref{sec:dec} introduces basic notations used later in the manuscript. Section \ref{sec:cond} describes a useful technical instrument, called conditioning, which is then used in Section \ref{sec:inf} and Section \ref{sec:samp} to describe algorithms for efficient inference and learning, respectively, of the zero-field Ising models over graphs, which allows for a $c$-nice decomposition, where $c$ is a positive integer. Section \ref{sec:k5} describes an application of the algorithm to the example of the $K_5$-free zero-field Ising models. Section \ref{sec:emp} presents an empirical application of the newly introduced family of tractable models to an upper-bounding log-partition function of a broader family of intractable graphical models (planar nonzero-field Ising models). Section \ref{sec:concl} is reserved for conclusions.

\section{Algorithm} \label{sec:main}

We commence by introducing the concept of $c$-nice decomposition of a graph and stating the main result on the tractability of the new family of Ising models in subsection \ref{sec:dec}. We introduce a helpful ``conditioning'' machinery in subsection \ref{sec:cond} and then describe the efficient inference (subsection \ref{sec:inf}) and sampling (subsection \ref{sec:samp}) algorithms which constructively prove the statement.


\subsection{Decomposition tree and the main result} \label{sec:dec}

Throughout the text, we use common graph-theoretic notations and definitions \cite{diestel} and also restate the most important concepts briefly. We mainly follow \cite{curticapean,reed} in the definition of the decomposition tree and its properties sufficient for our goals. Again, we point out that we only consider graphs without loops or multiple edges.

The graph is \textit{planar} when it can be drawn on a plane without edge intersections. Graph $G'$ is a \textit{subgraph} of $G$ whenever $V(G') \subseteq V(G)$ and $E(G') \subseteq E(G)$. For two subgraphs $G'$ and $G''$ of $G$, let $G' \cup G'' = (V(G') \cup V(G''), E(G') \cup E(G''))$ (graph \textit{union}).


Consider a tree decomposition $\mathcal{T} = \langle T, \mathcal{G} \rangle$ of a graph $G$ into a set of subgraphs $\mathcal{G} \triangleq \{ G_t \}$ of $G$, where $t$ are \textit{nodes} of a tree $T$, that is, $t \in V(T)$. One of the nodes of the tree, $r \in V(T)$, is selected as the root. For each node $t \in V(T)$, its \textit{parent} is the first node on the unique path from $t$ to $r$.
$G_{\leq t}$  denotes the graph union of $G_{t'}$ for all the nodes $t'$ in $V(T)$ that are $t$ or its descendants.
$G_{\nleq t}$ denotes the graph union of $G_{t'}$ for all the nodes $t'$ in $V(T)$ that are neither $t$ nor descendants of~$t$.
For two neighboring nodes of the tree, $t, p \in V(T)$ and $\{ t, p \} \in E(T)$, the set of overlapping vertices of $G_t$ and $G_p$, $K \triangleq V(G_t) \cap V(G_p)$, is called an  \textit{attachment set} of~$t$ or~$p$. If $p$ is a parent of~$t$, then $K$ is a \textit{navel} of $t$. We assume that the navel of the root is empty.

$\mathcal{T}$ is a \textit{$c$-nice decomposition} of $G$ if the following requirements are satisfied:
\begin{enumerate}
    \item $\forall t \in V(T)$ with a navel $K$, it holds that $K = V(G_{\leq t}) \cap V(G_{\nleq t})$.

    \item Every attachment set $K$ is of size $0$, $1$, $2$, or $3$.

    \item $\forall t \in V(T)$, either $| V(G_t) | \leq c$ or $G_t$ is planar.

    \item If $t \in V(T)$ is such that $| V(G_t) | > c$, addition of all edges of type $e = \{ v, w \}$, where $v, w$ belong to the same attachment set of $t$ (if $e$ is not yet in $E(G_t)$) does not destroy planarity of $G_t$.
\end{enumerate}


Stating it informally, the $c$-nice decomposition of $G$ is a tree decomposition of $G$ into planar and ``small'' (of size at most $c$) subgraphs $G_t$, ``glued'' via subsets of at most three vertices of $G$. Figure~\ref{fig:k5}(a) shows an example of a $c$-nice decomposition with $c = 8$.
There are various similar ways to define a graph decomposition in literature, and the one presented above is customized to include only properties significant for our consecutive analysis. The remainder of this section is devoted to a constructive proof of the following statement.
\begin{theorem} \label{th:main}
    Let $\mathcal{I} = \langle G, 0, J \rangle$ be any zero-field Ising model where there exists a $c$-nice decomposition $\mathcal{T}$ of $G$, where $c$ is an absolute constant. Then, there is an algorithm which, given $\mathcal{I}, \mathcal{T}$ as an input, does two things: (1) finds $Z$ and (2) samples a configuration from $\mathcal{I}$ in time $O( \sum\limits_{t \in V(T)} | V(G_t) |^\frac32 )$. 
\end{theorem}



\begin{figure}[!th]
    \centering
    \includegraphics[width=0.9\linewidth]{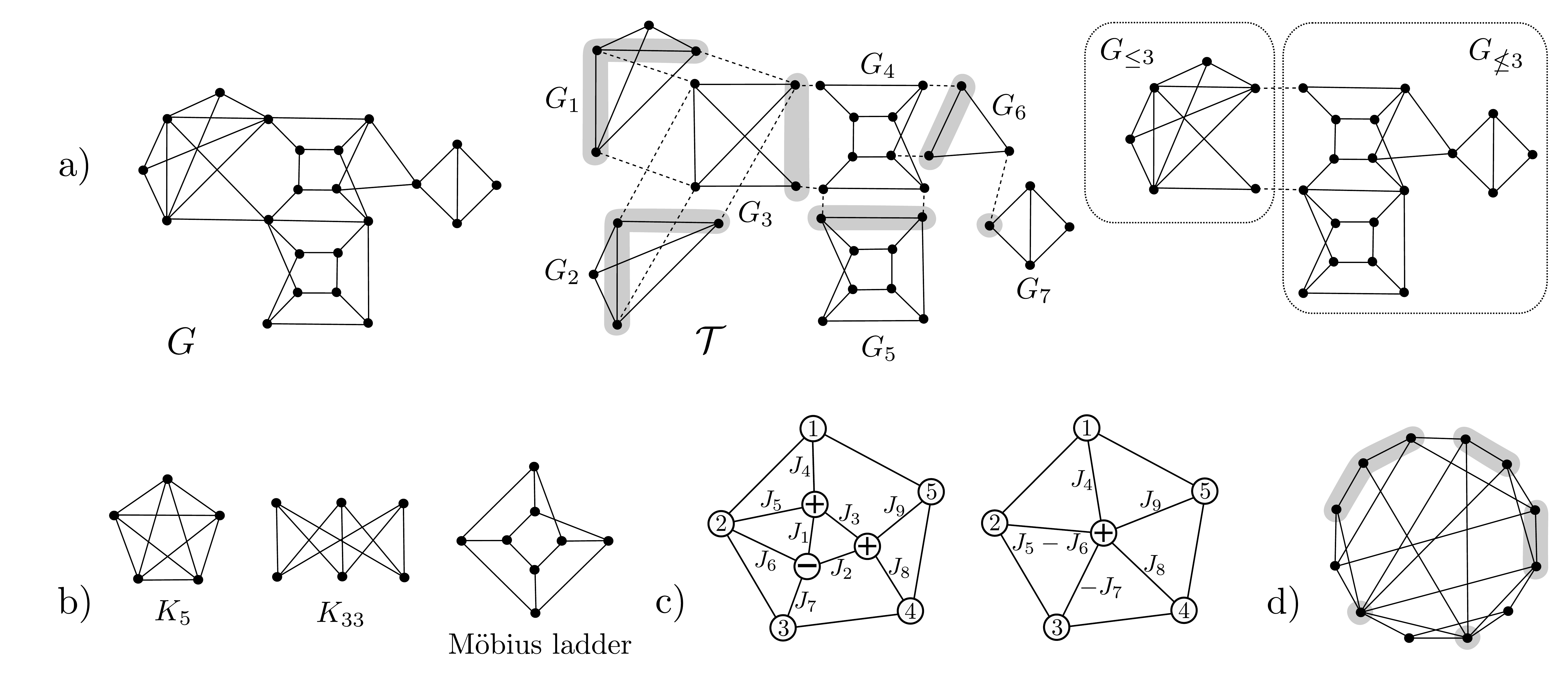}
    \caption{a) An exemplary graph $G$ and its $8$-nice decomposition $\mathcal{T}$, where $t\in \{1, \cdots, 7\}$ labels nodes of the decomposition tree $T$ and node $4$ is chosen as the root ($r = 4$). Identical vertices of $G$ in its subgraphs $G_t$ are shown connected by dashed lines. Navels of size $1$, $2$, and $3$ are highlighted. Component $G_5$ is nonplanar, and $G_4$ becomes nonplanar when all attachment edges are added (according to the fourth item of the definition of the $c$-nice decomposition). $G_{\leq 3}$ and $G_{\nleq 3}$ are shown with dotted lines. Note that the decomposition is non-unique for the graph. For instance, edges that belong  to the attachment set can go to either of the two subgraphs containing this set or even repeat in both. b) Minors $K_5$ and $K_{33}$ are forbidden in the planar graphs. M\"{o}bius ladder and its subgraphs are the only nonplanar graphs allowed in the $8$-nice decomposition of a $K_5$-free graph. c) The left panel is an example of conditioning on three vertices/spins in the center of a graph. The right panel shows a modified graph where the three vertices (from the left panel) are reduced to one vertex, then leading to a modification of the pairwise interactions within the associated zero-field Ising model over the reduced graph. 
    d) Example of a graph that contains $K_5$ as a minor: by contracting the highlighted groups of vertices and deleting the remaining vertices, one arrives at the $K_5$ graph.}
    \label{fig:k5}
\end{figure}

\subsection{Inference and sampling conditioned on 1, 2, or 3 vertices/spins} \label{sec:cond}

Before presenting the algorithm that proves Theorem \ref{th:main} constructively, let us introduce auxiliary machinery of ``conditioning'', which describes the partition function of a zero-field Ising model over a planar graph conditioned on $1$, $2$, or $3$ spins. 
Consider a zero-field Ising model $\mathcal{I} = \langle G, 0, J \rangle$ defined over a planar graph $G$. Recall the following result, rigorously proven in \cite[Corollary under Theorem 1]{k33},  which we intend to use in the aforementioned tree decomposition as a black box. 
\begin{theorem}
    Given $\mathcal{I} = \langle G, 0, J \rangle$, where $G$ is planar, $Z$ can be found in time $O(N^\frac32)$. Drawing a sample from $\mathcal{I}$ is a task of $O(N^\frac32)$ complexity.
    \label{th:planar}
\end{theorem}

Let us now introduce the notion of \textit{conditioning}. Consider a spin configuration $X \in \{ \pm 1 \}^N$, a subset $V' = \{ v^{(1)}, \dots, v^{(\omega)} \} \subseteq V(G)$, and define a \textit{condition} $S = \{ x_{v^{(1)}} = s^{(1)}, \dots, x_{v^{(\omega)}} = s^{(\omega)} \}$ \textit{on} $V'$, where $s^{(1)}, \dots, s^{(\omega)} = \pm 1$ are fixed values. Conditional versions of the probability distribution (\ref{eq:zfim}--\ref{eq:Z}) and the \textit{conditional} partition function become 
\begin{eqnarray}
    && \mathbb{P}(X | S) \triangleq \frac{\mathcal{W}(X) \times \mathbbm{1}(X | S)}{Z_{|S}}, \quad \mathbbm{1} (X | S) \triangleq 
    \left\{ \begin{array} {cc}1, & x_{v^{(1)}} = s^{(1)}, \dots, x_{v^{(\omega)}} = s^{(\omega)} \\ 0, & \mbox{otherwise} \end{array} \right.
, 
    \label{eq:zfim_cond} \\ &&
    \text{where}~Z_{|S} \triangleq \sum_{X \in \{\pm 1\}^N} \mathcal{W}(X) \times \mathbbm{1} (X | S). \label{eq:Z_cond}
\end{eqnarray}

Notice that when $\omega = 0$, $S = \{ \}$ and (\ref{eq:zfim_cond}--\ref{eq:Z_cond}) is reduced to (\ref{eq:zfim}--\ref{eq:Z}). The subset of $V(G)$ is \textit{connected} whenever the subgraph, induced by this subset is connected. Theorem \ref{th:planar} can be extended as follows (a formal proof can be found in the supplementary materials).
\begin{lemma} \label{lemma:cond}
Given $\mathcal{I} = \langle G, 0, J \rangle$ where $G$ is planar and a condition $S$ on a connected subset $V' \subseteq V(G)$, $| V' | \leq 3$, computing the conditional partition function $Z_{|S}$ and sampling from $\mathbb{P}(X | S)$
are tasks of $O(N^\frac32)$ complexity.
\end{lemma}

We omit here the tedious proof of the Lemma, also mentioning that the conditioning algorithm proving the Lemma takes the subset of connected vertices and ``collapses'' them into a single vertex. The graph remains planar and the task is reduced to conditioning on one vertex, which is an elementary operation given in Theorem \ref{th:planar}. (See Figure \ref{fig:k5}(c) for an illustration.)

\subsection{Inference algorithm} \label{sec:inf}


This subsection constructively proves the inference part of Theorem \ref{th:main}.
For each $t \in V(T)$, let $\mathcal{I}_{\leq t} \triangleq \langle G_{\leq t}, 0, \{ J_e \, | \, e \in E(G_{\leq t}) \subseteq E(G) \} \rangle$ denote a zero-field Ising \textit{submodel} \textit{induced} by $G_{\leq t}$. Denote the partition function and subvector of $X$ related to $\mathcal{I}_{\leq t}$ as $Z^{\leq t}$ and $X_{\leq t} \triangleq \{ x_v | v \in V(G_{\leq t}) \}$, respectively.

Further, let $K$ be $t$'s navel and let $S = \{ \forall v \in K: x_v = s^{(v)} \}$ denote some condition on $K$.
Recall that $| K | \leq 3$.
For each $t$, the algorithm computes conditional partition functions $Z^{\leq t}_{| S}$ for all choices of condition spin values $\{ s^{(v)} = \pm 1 \}$.
Each $t$ is processed only when its children have already been processed, so the algorithm starts at the leaf and ends at the root. If $r \in G(T)$ is a root, its navel is empty and $G_{\leq r} = G$, hence $Z = Z^{\leq r}_{| \{  \}}$ is computed after $r$'s processing.

Suppose all children of $t$, $c_1, ..., c_m \in V(T)$ with navels $K_1, ..., K_m \subseteq V(G_t)$ have already been processed, and now $t$ itself is considered.  Denote a spin configuration on $G_t$ as $Y_t \triangleq \{ y_v = \pm 1 \, | \, v \in V(G_t) \}$.
$\mathcal{I}_{\leq c_1}, ..., \mathcal{I}_{\leq c_m}$ are $\mathcal{I}_{\leq t}$'s submodels induced by $G_{\leq c_1}, ..., G_{\leq c_m}$, which can only intersect at their navels in $G_t$. Based on this, one states the following dynamic programming relation:
\begin{equation}
Z^{\leq t}_{| S} = \sum_{Y_t \in \{ \pm 1 \}^{| V(G_t) |} } \mathbbm{1} (Y_t | S) \exp \left( \sum_{e = \{ v, w \} \in E(G_t)} J_e y_v y_w \right) \cdot \prod_{i = 1}^m 
Z^{\leq c_i}_{| S_i [ Y_t ]}.
    \label{eq:dp}
\end{equation}
Here, $S_i [ Y_t ]$ denotes a condition $\{ \forall v \in K_i: x_v = y_v \}$ on $K_i$. The goal is to efficiently perform summation in~\eqref{eq:dp}. Let $I^{(0)}, I^{(1)}, I^{(2)}, I^{(3)}$ be a partition of $\{ 1, ..., m \}$ by navel sizes. Figure \ref{fig:infsamp}(a,b) illustrates inference in $t$.

\begin{enumerate}[wide, labelwidth=!, labelindent=0pt]
    \item \textbf{Navels of size 0, 1.} Notice that if $i \in I^{(0)}$, then $Z^{\leq c_i}_{| \{  \}} = Z^{\leq c_i}$ is constant, which was computed before. The same is true for $i \in I^{(1)}$ and $Z^{\leq c_i}_{| S^{(i)} [Y_t]} = \frac{1}{2} Z^{\leq c_i}$.
%
    \item \textbf{Navels of size 2.} Let $i \in I^{(2)}$ denote $K_i = \{ u^i, q^i \}$ and simplify notation $Z^{\leq c_i}_{y_1, y_2} \triangleq Z^{\leq c_i}_{x_{u^i} = y_1, x_{q^i} = y_2}$ for convenience. Notice that $Z^{\leq c_i}_{|S_i [ Y_t ]}$ is strictly positive, and due to the zero-field nature of $\mathcal{I}_{\leq c_i}$, one finds $Z^{\leq c_i}_{| +1, +1} = Z^{\leq c_i}_{| -1, -1}$ and $Z^{\leq c_i}_{| +1, -1} = Z^{\leq c_i}_{| -1, +1}$. Then, one arrives at $\log Z^{\leq c_i}_{| S_i [ Y_t ]} = A_i + B_i y_{u^i} y_{q^i} $, where $A_i \triangleq \log Z^{\leq c_i}_{| +1, +1} + \log Z^{\leq c_i}_{| +1, -1}$ and $B_i \triangleq \log Z^{\leq c_i}_{| +1, +1} - \log Z^{\leq c_i}_{| +1, -1}$.
    \item \textbf{Navels of size 3.} Let $i \in I^{(3)}$, and as above, denote $K_i = \{ u^i, q^i, h^i \}$ and $Z^{\leq c_i}_{y_1, y_2, y_3} \triangleq Z^{\leq c_i}_{x_{u^i} = y_1, x_{q^i} = y_2, x_{h^i} = y_3}$. 
   \begin{figure}
    \centering
    \includegraphics[width=\linewidth]{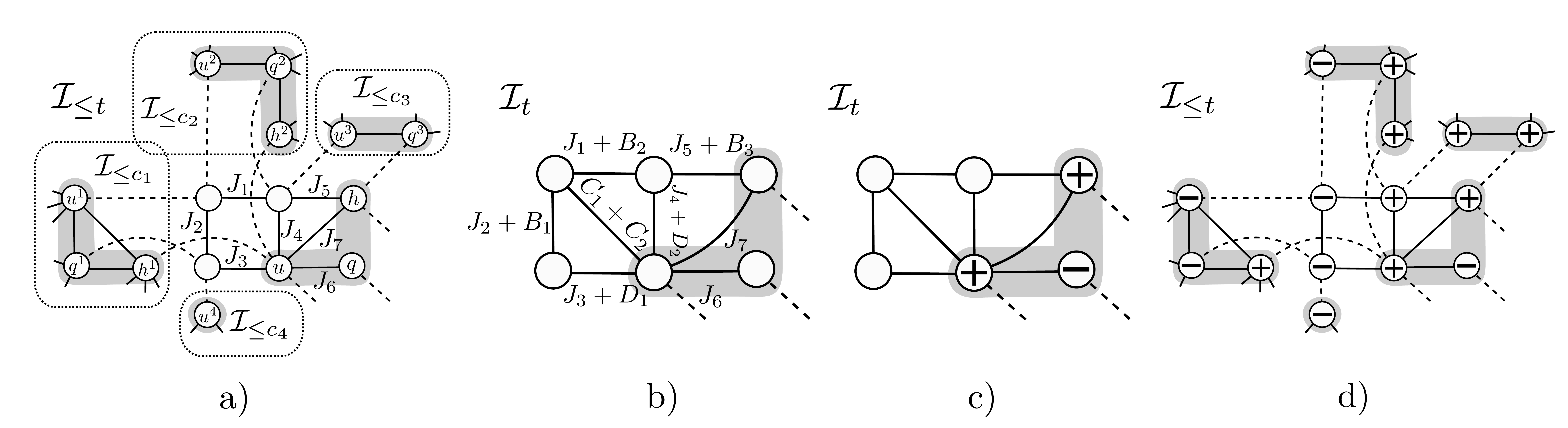}
    \caption{a) Example of inference at node $t$ with children $c_1, c_2, c_3, c_4$. Navels $K_1 = \{ u^1, q^1, h^1 \}, K_2 = \{ u^2, q^2, h^2 \}, K_3 = \{ u^2, q^2 \}, K_4 = \{ u^4 \}$, and $K = \{ u, q, h \}$ are highlighted. Fragments of $\mathcal{I}_{\leq c_i}$ are shown with dotted lines. Here, $I^{(0)} = \varnothing, I^{(1)} = \{ 4 \}, I^{(2)} = \{ 3 \}, \text{and}~ I^{(3)} = \{ 1, 2 \}$, indicating that one child is glued over one node, one child is glued over two nodes, and two children are glued over three nodes. b) ``Aggregated'' Ising model $\mathcal{I}_t$ and its pairwise interactions are shown. Both c) and d) illustrate sampling over $\mathcal{I}_t$. One sample spins in $\mathcal{I}_t$ conditioned on $S^{(t)}$ and then repeats the procedure at the child nodes.}
    \label{fig:infsamp}
\end{figure} 
    
    Due to the zero-field nature of $\mathcal{I}_{\leq c_i}$, it holds that $Z^{\leq c_i}_{| +1, y_2, y_3} = Z^{\leq c_i}_{| -1, y_2, y_3}$. Observe that there are such $A_i, B_i, C_i, D_i$ that $ \log Z^{\leq c_i}_{| y_1, y_2, y_3} = A_i + B_i 
    y_1 y_2+ C_i y_1 y_3 + D_i y_2 y_3$ for all $y_1, y_2, y_3=\pm 1$, which is guaranteed since the following system of equations has a solution:
    \begin{equation}
        \begin{bmatrix} \log Z^{\leq c_i}_{|+1, +1, +1} \\ \log Z^{\leq c_i}_{|+1, +1, -1 } \\ \log Z^{\leq c_i}_{|+1, -1, +1} \\ \log Z^{\leq c_i}_{|+1, -1, -1} \end{bmatrix} = \begin{bmatrix} +1 & +1 & +1 & +1 \\ +1 & +1 & -1 & -1 \\ +1 & -1 & +1 & -1 \\ +1 & -1 & -1 & +1 \end{bmatrix} \times \begin{bmatrix} A_i \\ B_i \\ C_i \\ D_i \end{bmatrix}.
        \label{eq:3vsystem}
    \end{equation}
\end{enumerate}

Considering three cases, one rewrites Eq.~(\ref{eq:dp}) as
\begin{align}
    Z^{\leq t}_{| S} &= M \cdot \sum_{Y_t} \mathbbm{1} (Y_t | S) \exp \biggl( \sum_{e = \{ v, w \} \in E(G_t)} J_e y_v y_w + \sum_{i \in I^{(2)} \cup I^{(3)}} B_i y_{u^i} y_{q^i} \nonumber \\
    &+ \sum_{i \in I^{(3)}} ( C_i y_{u^i} y_{h^i} + D_i y_{q^i} y_{h^i}) \biggr),
    \label{eq:dp4}
\end{align}
where $M \triangleq 2^{- | I^{(1)} |} \cdot \left( \prod_{i \in I^{(0)} \cup I^{(1)}} Z^{\leq c_i} \right) \cdot \exp( \sum_{i \in I^{(2)} \cup I^{(3)}} A_i)$. The sum in Eq.~(\ref{eq:dp4}) is simply a conditional partition function of a zero-field Ising model $\mathcal{I}_t$ defined over a graph $G_t$ with pairwise interactions of $\mathcal{I}$ adjusted by the addition of $B_i, C_i, \text{and}~ D_i$ summands at the appropriate navel edges (if a corresponding edge is not present in $G_t$, it has to be added).
If $| V(G_t) | \leq c$, then (\ref{eq:dp4}) is computed a maximum of four times (depending on navel size) by brute force ($O(1)$ time). Otherwise, if $K$ is a disconnected set in $G_t$, we add zero-interaction edges inside it to make it connected. Possible addition of edges inside $K, K_1, \dots, K_m$ doesn't destroy planarity according to the fourth item in the definition of the $c$-nice decomposition above. Finally, we compute (\ref{eq:dp4}) using Lemma \ref{lemma:cond} in time $O( | V(G_t) |^\frac32 )$.


The inference part of Theorem \ref{th:main} follows directly from the procedure just described.

\begin{figure}
    \centering
    \includegraphics[width=0.9\linewidth]{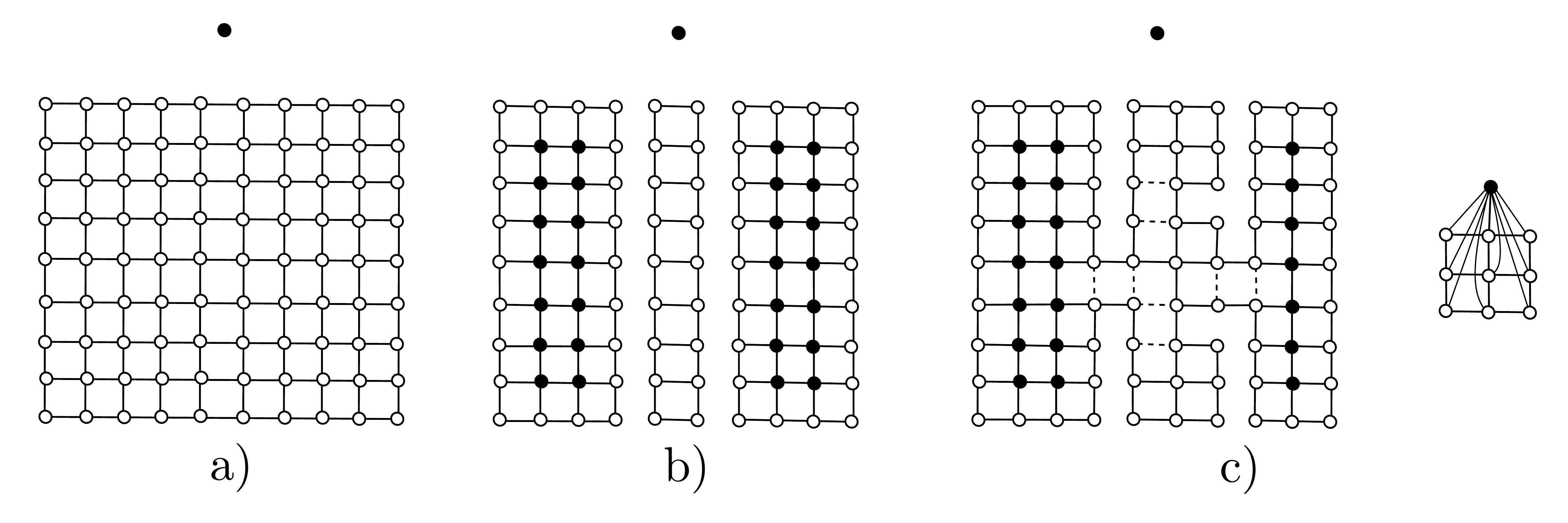}
    \caption{Construction of graphs used for approximate inference on a rectangular lattice. For better visualization, vertices connected to an apex are colored white. a) $G'$ graph. b) One of planar $G^{(r)}$ graphs used in \protect\cite{globerson}. Such ``separator'' pattern is repeated for each column
    and row, resulting in $2(H - 1)$ graphs in $\{ G^{(r)} \}$. In addition, \protect\cite{globerson}
    adds an \textit{independent variables} graph where only apex edges are drawn. c) A modified ``separator'' pattern we propose. 
    Again, the pattern is repeated horizontally and vertically resulting in $2(H - 2)$ graphs $+$ independent variables graph. This pattern covers more magnetic fields and connects separated parts. Dashed edges indicate the structure of $10$-nice decomposition used for inference. (Nonplanar node of size $10$ is illustrated on the right.)
    }
    \label{fig:grid}
\end{figure}

\subsection{Sampling algorithm} \label{sec:samp}

Next, we address the sampling part of Theorem \ref{th:main}.
We extend the algorithm from section \ref{sec:inf} so that it supports efficient sampling from $\mathcal{I}$.
Assume that the inference pass through $T$ (from leaves to root) has been done so that $\mathcal{I}_t$ for all $t \in V(T)$ are computed. 
Denote $X_t \triangleq \{ x_v \, | \, v \in V(G_t) \}$.
The sampling algorithm runs backwards, first drawing spin values $X_r$ at the root $r$ of $T$ from the marginal distribution $\mathbb{P}(X_r)$, and then processing each node $t$ of $T$ after its parent $p$ is processed. Processing consists of drawing spins $X_t$ from $\mathbb{P}(X_t \, | \, X_p) = \mathbb{P}(X_t \, | \, X^{(t)} \triangleq \{ x_v \, | \, v \in K \})$,
where $K$ is a navel of $t$. This marginal-conditional scheme generates the correct sample $X$ of spins over $G$.

Let 
$\mathbb{P}_{\leq t} (X_{\leq t})$ define a spin distribution of $\mathcal{I}_{\leq t}$. Because the Ising model is an example of Markov Random Field, it holds that $\mathbb{P}_{\leq t} (X_{\leq t} \, | \, X^{(t)} ) = \mathbb{P} (X_{\leq t} \, | \, X^{(t)} )$. We further derive 
\begin{align}
    &\mathbb{P}(X_t \, | \, X^{(t)}) = \mathbb{P}_{\leq t}(X_t \, | \, X^{(t)}) = \frac{1}{Z^{\leq t}} \sum_{X_{\leq t}\setminus X_t} \exp \biggl( \sum_{e = \{ v, w \} \in E(G_{\leq t})} J_e x_v x_w \biggr) \nonumber \\
    &= \frac{1}{Z^{\leq t}} \cdot \exp \biggl( \sum_{e = \{ v, w \} \in E(G_t)} J_e x_v x_w \biggr) \cdot \prod_{i = 1}^m Z^{\leq c_i}_{|S_i [X_t]} \nonumber \\
    &\propto \exp \biggl( \sum_{e = \{ v, w \} \in E(G_t)} J_e x_v x_w + \sum_{i \in I^{(2)} \cup I^{(3)}} B_i x_{u^i} x_{q^i} + \sum_{i \in I^{(3)}} ( C_i x_{u^i} x_{h^i} + D_i x_{q^i} x_{h^i}) \biggr). 
\end{align}

In other words, sampling from $\mathbb{P}(X_t \, | \, X^{(t)})$ is reduced to sampling from $\mathcal{I}_t$ conditional on spins $X^{(t)}$ in the navel $K$. It is done via brute force if $| V(G_t) | \leq c$; otherwise, Lemma \ref{lemma:cond} allows one to draw $X_t$ in $O(| V(G_t) |^\frac32)$, since $| K | \leq 3$. Sampling efforts cost as much as inference, which concludes the proof of Theorem \ref{th:main}. Figure \ref{fig:infsamp}(c,d) illustrates sampling in $t$.

\begin{figure}
    \centering
    \includegraphics[width=0.9\linewidth]{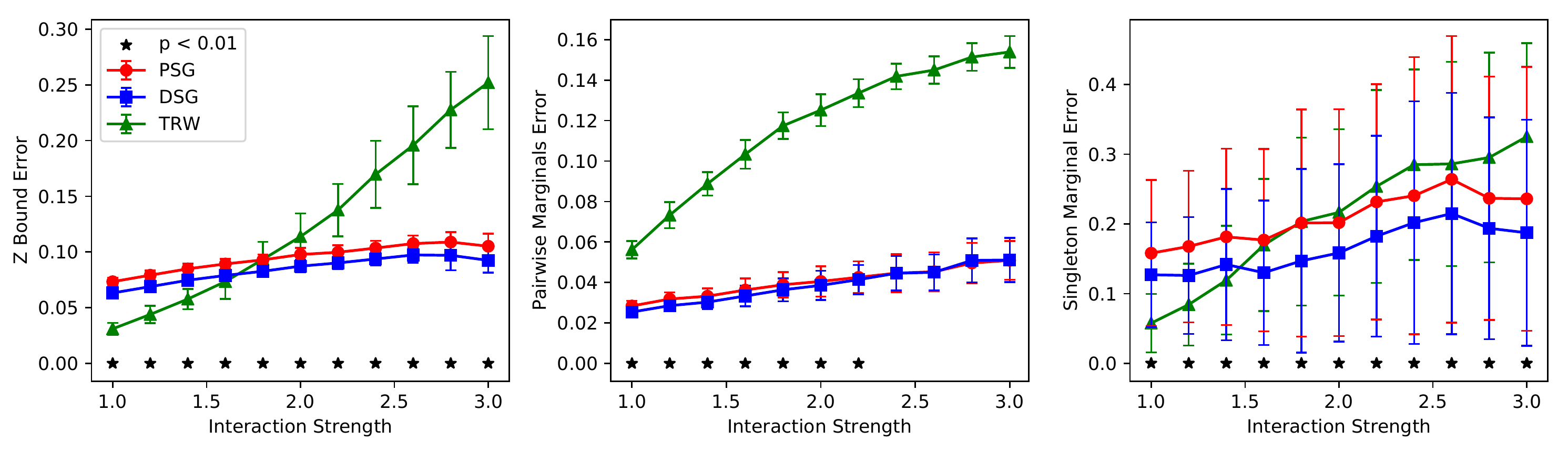}
    \caption{Comparison of tree-reweighted approximation (TRW), planar spanning graph (PSG), and decomposition-based spanning graph (DSG) approaches. The first plot is for normalized log-partition error, the second is for error in pairwise marginals, and the third is for error in singleton central marginal. Standard errors over $100$ trials are shown as error bars. An asterisk ``*'' 
    indicates the statistically significant improvement of DSG over PSG, with a $p$-value smaller than $0.01$ according to the Wilcoxon test with the Bonferroni correction \protect\cite{wilcoxon}. 
    }
    \label{fig:plot}
\end{figure}

\section{Application: \texorpdfstring{$K_5$}--free zero-field Ising models} \label{sec:k5}

\textit{Contraction} is an operation of removing two adjacent vertices $v$ and $u$ (and all edges incident to them) from the graph and adding a new vertex $w$ adjacent to all neighbors of $v$ and $u$.
For two graphs $G$ and $H$, $H$ is $G$'s \textit{minor}, if it is isomorphic to a graph obtained from $G$'s subgraph by a series of contractions. $G$ is \textit{$H$-free}, if $H$ is not $G$'s minor. We especially focus on the case when $H = K_5$ (Figure \ref{fig:k5}(d)). Planar graphs are a special type of $K_5$-free graphs, according to Wagner's theorem \cite[Chapter 4.4]{diestel}. Moreover, some nonplanar graphs are $K_5$-free, for example, $K_{33}$ (Figure \ref{fig:k5}(b)). $K_5$-free graphs are neither genus-bounded (a disconnected set of $g$ $K_{33}$ graphs is $K_5$-free and has a genus of $g$ \cite{battle}) and is not treewidth-bounded (planar square grid of size $t \times t$ is $K_5$-free and has a treewidth of $t$ \cite{bodlaender}).

\begin{theorem} \label{th:k5dec}
Let $G$ be a $K_5$-free graph of size $N$ with no loops or multiple edges. Then, the $8$-nice decomposition $\mathcal{T}$ of $G$ exists and can be computed in time $O(N)$.
\end{theorem}
\begin{proof}[Proof (Sketch).]
    An equivalent decomposition is constructed in \cite{reed} in time $O(N)$. We put a formal proof in the Supplementary materials.
\end{proof}
\begin{remark}
    The $O(N)$ construction time of $\mathcal{T}$ guarantees that 
    $\sum_{t \in V(T)} | V(G_t) | = O(N)$. All nonplanar components in $\mathcal{T}$ are isomorphic to the M\"{o}bius ladder (Figure \ref{fig:k5}(b)) or its subgraph.
\end{remark}

The graph in Figure \ref{fig:k5}(a) is actually $K_5$-free. Theorems \ref{th:main} and \ref{th:k5dec} allow us to conclude:
\begin{theorem} \label{th:inf}
    Given $\mathcal{I} = \langle G, 0, J \rangle$ with $K_5$-free $G$ of size $N$, finding $Z$ and sampling from $\mathcal{I}$ take $O(N^\frac32)$ total time.
\end{theorem}
\begin{proof}
    Finding $8$-nice $\mathcal{T}$ for $G$ takes $O(N)$ time (Theorem \ref{th:k5dec}). Provided with $\mathcal{T}$, the complexity is
    \begin{equation*}
        O \left(\sum_{t \in V(T)} | V(G_t) |^\frac32 \right) = O \left( (\sum_{t \in V(T)} | V(G_t) | )^\frac32 \right) = O(N^\frac32),
    \end{equation*}
    where we apply convexity of $f(z) = z^\frac32$ for $z \geq 0$ and the Remark after Theorem \ref{th:k5dec}.
\end{proof}


\section{Application: approximate inference of square-grid Ising model} \label{sec:emp}

In this section, we consider $\mathcal{I} = \langle G, \mu, J \rangle$ such that $G$ is a square-grid graph of size $H \times H$. Finding $Z (G, \mu, J)$ for arbitrary $\mu$, $J$ is an NP-hard problem \cite{barahona} in such a setting. Construct $G'$ by adding an \textit{apex} vertex connected to all $G$'s vertices by edge (Figure \ref{fig:grid}(a)). Now it can easily be seen that $Z (G, \mu, J) = \frac{1}{2} Z(G', 0, J' = (J_\mu \cup J))$, where $J_\mu = \mu$ are interactions assigned for apex edges.

Let $\{ G^{(r)} \}$ be a family of spanning graphs ($V(G^{(r)}) = V(G')$, $E(G^{(r)}) \subseteq E(G')$) and $J^{(r)}$ be interaction values on $G^{(r)}$. Also, denote $\hat{J}^{(r)} = J^{(r)} \cup \{ 0, e \in E(G') \setminus E(G^{(r)}) \}$. Assuming that $\log Z (G^{(r)}, 0, J^{(r)})$ are tractable, the convexity of $\log Z (G', 0, J')$ allows one to write the following upper bound:
\begin{equation} \label{eq:ub}
    \log Z (G', 0, J') \leq \min_{\substack{\rho(r) \geq 0, \sum_r \rho(r) = 1 \\ \{ J^{(r)} \}, \sum_r \rho(r) \hat{J}^{(r)} = J'}} \sum_r \rho(r) \log Z(G^{(r)}, 0, J^{(r)}).
\end{equation}

After graph set $\{ G^{(r)} \}$ has been fixed, one can numerically optimize the right-hand side of (\ref{eq:ub}), as shown in \cite{globerson} for planar $G^{(r)}$.  The extension of the basic planar case is straightforward and can be found in the Supplementary materials for convenience. We also put the description of marginal probabilities approximation suggested in \cite{globerson,wainwright}.

The choice for a planar spanning graph (PSG) family $\{ G^{(r)} \}$ of \cite{globerson} is illustrated in Figure \ref{fig:grid}(b). A tractable decomposition-based extension of the planar case presented in this manuscript suggests a more advanced construction---decomposition-based spanning graphs (DSG) (Figure \ref{fig:grid}(c)). We compare performance of both PSG and DSG approaches as well as the performance of tree-reweighted approximation (TRW) \cite{wainwright} in the following setting of \textit{Varying Interaction}: $\mu \sim \mathcal{U}(-0.5, 0.5)$, $J \sim \mathcal{U}(-\alpha, \alpha)$, where $\alpha \in \{ 1, 1.2, 1.4, \dots, 3 \}$. We opt optimize for grid size $H = 15$ ($225$ vertices, $420$ edges) and compare upper bounds and marginal probability approximations (superscript \textit{alg}) with exact values obtained using a junction tree algorithm \cite{jensen} (superscript \textit{true}). We compute three types of error:
\begin{enumerate}
    \item normalized log-partition error $\frac{1}{H^2} (\log Z^{alg} - \log Z^{true})$,
    \item error in pairwise marginals $\frac{1}{| E(G) |} \sum_{e = \{ v, w \} \in E(G)} | \mathbb{P}^{alg} (x_v x_w = 1) - \mathbb{P}^{true} (x_v x_w = 1) |$, and
    \item error in singleton central marginal $| \mathbb{P}^{alg} (x_v = 1) - \mathbb{P}^{true} (x_v = 1)|$ where $v$ is a vertex of $G$ with coordinates $(8, 8)$.
\end{enumerate}

We average results over $100$ trials (see Fig.~\ref{fig:plot}).\footnote{Hardware used: 24-core Intel\textregistered \, Xeon\textregistered \, Gold 6136 CPU @ 3.00 GHz}\footnote{Implementation of the algorithms is available at \url{https://github.com/ValeryTyumen/planar\_ising}} We use the same quasi-Newton algorithm \cite{bertsekas} and parameters when optimizing (\ref{eq:ub}) for PSG and DSG, but for most settings, DSG outperforms PSG and TRW. Cases with smaller TRW error can be explained by the fact that TRW implicitly optimizes~\eqref{eq:ub}  over the family of \textit{all} spanning trees which can be exponentially big in size, while for PSG and DSG we only use $O(H)$ spanning graphs.

Because PSG and DSG approaches come close to each other, we additionally test for each value of $\alpha$ on each plot, whether the difference $err_{PSG} - err_{DSG}$ is bigger than zero. We apply a one-sided Wilcoxon's test \cite{wilcoxon} together with the Bonferroni correction because we test $33$ times \cite{bonferroni}. In most settings, the improvement is statistically significant (Figure \ref{fig:plot}). 


\section{Conclusion} \label{sec:concl}

In this manuscript, we introduce a new family of zero-field Ising models composed of planar components and graphs of $O(1)$ size. For these models, we describe a polynomial algorithm for exact inference and sampling provided that the decomposition tree is also in the input. 
A theoretical application is $O(N^\frac32)$ inference and sampling algorithm for $K_5$-free zero-field Ising models---a superset of planar zero-field models that is neither treewidth- nor genus-bounded.
A practical application is an improvement of an approximate inference scheme for arbitrary topologies based on planar spanning graphs \cite{globerson} but using tractable spanning decomposition-based graphs instead of planar. We leave the algorithm as it is but substitute planar graphs with a family of spanning decomposition-based graphs that are tractable. This alone gives a tighter upper bound on the true partition function and a more precise approximation of marginal probabilities.

\section{Acknowledgements}
This work was supported by the U.S. Department of Energy through the Los Alamos National Laboratory as part of LDRD and the DOE Grid Modernization Laboratory Consortium (GMLC). Los Alamos National Laboratory is operated by Triad National Security, LLC, for the National Nuclear Security Administration of U.S. Department of Energy (Contract No. 89233218CNA000001).

\bibliographystyle{chicago}
\bibliography{arxiv}

\newpage
\begin{appendices}
\section*{
Proof for Lemma \ref{lemma:cond}}

\begin{customlm}{\ref{lemma:cond}}
    Given $\mathcal{I} = \langle G, 0, J \rangle$ where $G$ is planar and a condition $S$ on a connected subset $V' \subseteq V(G)$, $| V' | \leq 3$, computing conditional partition function $Z_{|S}$, and sampling from $\mathbb{P}(X | S)$ are tasks of $O(N^\frac32)$ complexity.
\end{customlm}

\begin{proof} We consider cases depending on $\omega$ and consequently reduce each case to a simpler one. For convenience in cases where applies we denote $u \triangleq v^{(1)}, h \triangleq v^{(2)}, q \triangleq v^{(3)}$:

\begin{enumerate}

\item \textbf{Conditioning on $\omega = 0$ spins.} See Theorem \ref{th:planar}.

\item \textbf{Conditioning on $\omega = 1$ spin.} Since configurations $X$ and $-X$ have the same probability in $\mathcal{I}$, one deduces that $Z_{\,| \, x_u = s^{(1)}} = \frac12 Z$.

One also deduces that sampling $X$ from $\mathbb{P}(X\, | \, x_u = s^{(1)})$ is reduced to 1) drawing $\overline{X} = \{ \overline{x}_v = \pm 1 \}$ from $\mathbb{P}(X)$ and then 2) returning $X = (s^{(1)} \overline{x}_u) \cdot \overline{X}$ as a result.

\item \textbf{Conditioning on $\omega = 2$ spins.} There is an edge $e^{0} = \{ u, h \} \in E(G)$. The following expansion holds:
\begin{align}
    Z_{\,| \, x_u = s^{(1)}, x_h = s^{(2)}} &= \sum_{X, \, x_u = s^{(1)}, \, x_h = s^{(2)}} \exp\bigl( \sum_{ e = \{ v, w \} \in E(G)} J_e x_v x_w \bigr) \nonumber \\
    & = \exp (J_{e^0} s^{(1)} s^{(2)}) \cdot \sum_{X, \, x_u = s^{(1)}, \, x_h = s^{(2)}} \exp\bigl( \sum_{\substack{e = \{ v, w \} \in E(G) \\ e \neq e^0}} J_e x_v x_w\bigr) \nonumber \\
    &= \exp (J_{e^0} s^{(1)} s^{(2)}) \cdot \sum_{X, \, x_u = s^{(1)}, \, x_h = s^{(2)}} \exp\bigl( \sum_{\substack{e = \{ v, w \} \in E(G) \\ e \cap e^0 = \varnothing}} J_e x_v x_w \nonumber \\
    & + \sum_{\substack{e = \{ u, v \} \in E(G) \\ v \neq h}} (J_e s^{(1)}) x_v \cdot 1 + \sum_{\substack{e = \{ h, v \} \in E(G) \\ v \neq u}} (J_e s^{(2)}) x_v \cdot 1\bigr) \label{eq:cond2v}
\end{align}

Obtain graph $G'$ from $G$ by contracting $u, h$ into $z$. $G'$ is still planar and has $N - 1$ vertices. Preserve pairwise interactions of edges which were not deleted after contraction. For each edge $e = \{u, v\}$, $v \neq h$ set $J_{\{ z, v \}} = J_e s^{(1)}$, for each edge $e = \{ h, v \}$, $v \neq u$ set $J_{\{ z, v \}} = J_e s^{(2)}$. Collapse double edges in $G'$ which were possibly created by transforming into single edges. A pairwise interaction of the result edge is set to the sum of collapsed interactions.

Define a zero-field Ising model $\mathcal{I}'$ on the resulted graph $G'$ with its pairwise interactions, inducing a distribution $\mathbb{P}' (X' = \{ x'_v = \pm 1 | v \in V(G') \})$. Let $Z'$ denote $\mathcal{I}'$'s partition function. A closer look at (\ref{eq:cond2v}) reveals that
\begin{equation}
    Z_{\, | \, x_u = s^{(1)}, x_h = s^{(2)}} = \exp (J_{e^0} s^{(1)} s^{(2)}) \cdot Z'_{\,|\, x'_z = 1}
    \label{eq:ztrans}
\end{equation}
where $Z'_{\,|\, z'_y = 1}$ is a partition function conditioned on a single spin and can be found efficiently as shown above.

Since the equality of sums (\ref{eq:ztrans}) holds summand-wise, for a given $X'' = \{x''_v = \pm 1 \, | \, v \in V(G) \setminus \{ u, h \} \}$ the probabilities $\mathbb{P}(X'' \cup \{ x_u = s^{(1)}, x_h = s^{(2)} \} \, | \, x_u = s^{(1)}, x_h = s^{(2)})$ and $\mathbb{P}'(X'' \cup \{ x'_z = 1 \} \, | \, x'_z = 1)$ are the same. Hence, sampling from $\mathbb{P}(X \, | \, x_u = s^{(1)}, x_h = s^{(2)})$ is reduced to conditional sampling from planar zero-field Ising model $\mathbb{P}'(X' \, | \, x'_z = 1)$ of size~$N - 1$.

\item \textbf{Conditioning on $w = 3$ spins.} Without loss of generality assume that $u, h$ are connected by an edge $e^0$ in $G$. A derivation similar to (\ref{eq:cond2v}) and (\ref{eq:ztrans}) reveals that (preserving the notation of Case 2)
\begin{equation}
    Z_{\, | \, x_u = s^{(1)}, x_h = s^{(2)}, x_q = s^{(3)}} = \exp (J_{e^0} s^{(1)} s^{(2)}) \cdot Z'_{\,|\, x'_z = 1, x'_q = s^{(3)}}
\end{equation}
which reduces inference conditional on $3$ vertices to a simpler case of 2 vertices. Again, sampling from $\mathbb{P}(X \, | \, x_u = s^{(1)}, x_t = s^{(2)}, x_q = s^{(3)})$ is reduced to a more basic sampling from $\mathbb{P}'(X' \, | \, x'_z = 1, x'_q = s^{(3)})$.
\end{enumerate}
\end{proof}
In principle, Lemma \ref{lemma:cond} can be extended to arbitrarily large $\omega$ leaving a certain freedom for the Ising model conditioning framework. However, in this manuscript we focus on a given special case which is enough for our goals.

\section*{
Proof for Theorem \ref{th:k5dec}}

Prior to the proof, we introduce a series of definitions used in \cite{reed}. It is assumed that a graph $G = (V, E)$ (no loops and multiple edges) is given.

For any $X \subseteq V(G)$ let $G - X$ denote a graph $(V(G) \setminus X, \{ e = \{ v, w \} \in E(G) \, | \, v, w \notin X \})$. $X \subseteq V(G)$ is a $(i, j)$-\textit{cut} whenever $| X | = i$ and $G - X$ has at least $j$ connected components.

The graph is \textit{biconnected} whenever it has no $(1, 2)$-cut. A \textit{biconnected component} of the graph is a maximal biconnected subgraph. Clearly, a pair of biconnected components can intersect in at most one vertex and a graph of components' intersections is a tree when $G$ is connected (\textit{a tree of biconnected components}). The graph is \textit{$3$-connected} whenever it has no $(2, 2)$-cut.

A \textit{$2$-block tree} of a biconnected graph $G$, written $\langle T', \mathcal{G}' \rangle$, is a tree $T'$ with a set $\mathcal{G}' = \{ G'_t \}_{t \in V(T')}$ with the following properties:
\begin{itemize}
    \item[--] $G'_t$ is a graph (possibly with multiple edges) for each $t \in V(T')$.
    \item[--] If $G$ is $3$-connected then $T'$ has a single node $r$ which is colored $1$ and $G'_r = G$.
    \item[--] If $G$ is not $3$-connected then there exists a color $2$ node $t \in V(T')$ such that
    \begin{enumerate}
        \item $G'_t$ is a graph with two vertices $u$ and $v$ and no edges for some $(2, 2)$-cut $\{ u, v \}$ in $G$.
        \item Let $T'_1, \dots, T'_k$ be the connected components (subtrees) of $T' - t$. Then $G - \{ u, v \}$ has $k$ connected components $U_1, \dots, U_k$ and there is a labelling of these components such that $T'_i$ is a $2$-block tree of $G'_i = (V(U_i) \cup \{ u, v \}, E(U_i) \cup \{ \{ u, v \} \})$.
        \item For each $i$, there exists exactly one color $1$ node $t_i \in V(T'_i)$ such that $\{ u, v \} \subseteq V(G'_{t_i})$.
        \item For each $i$, $\{ t, t_i \} \in E(T)$.
    \end{enumerate}
\end{itemize}

A \textit{$(3, 3)$-block tree} of a $3$-connected graph $G$, written $\langle T'', \mathcal{G}'' \rangle$, is a tree $T''$ with a set $\mathcal{G}'' = \{ G''_t \}_{t \in V(T'')}$ with the following properties:
\begin{itemize}
    \item[--] $G''_t$ is a graph (possibly with multiple edges) for each $t \in V(T'')$.
    \item[--] If $G$ has no $(3, 3)$-cut then $T$ has a single node $r$ which is colored $1$ and $G_r = G$.
    \item[--] If $G$ has a $(3, 3)$-cut then there exists a color $2$ node $t \in V(T'')$ such that
    \begin{enumerate}
        \item $G''_t$ is a graph with vertices $u$, $v$ and $w$ and no edges for some $(3, 3)$-cut $\{ u, v, w \}$ in $G$.
        \item Let $T''_1, \dots, T''_k$ be the connected components (subtrees) of $T'' - t$. Then $G - \{ u, v, w \}$ has $k$ connected components $U_1, \dots, U_k$ and there is a labelling of these components such that $T_i$ is a $(3, 3)$-block tree of $G''_i = (V(U_i) \cup \{ u, v, w \}, E(U_i) \cup \{ \{ u, v \}, \{ v, w \}, \{ u, w \} \})$.
        \item For each $i$, there exists exactly one color $1$ node $t_i \in V(T''_i)$, such that $\{ u, v, w \} \subseteq V(G''_{t_i})$.
        \item For each $i$, $\{ t, t_i \} \in E(T'')$.
    \end{enumerate}
\end{itemize}

\begin{customthm}{\ref{th:k5dec}}
Let $G$ be $K_5$-free graph of size $N$ with no loops or multiple edges. Then the $8$-nice decomposition $\mathcal{T}$ of $G$ exists and can be computed in time $O(N)$.
\end{customthm}

\begin{proof}
Since $G$ is $K_5$-free and has no loops or multiple edges, it holds that $| E(G) | = O(N)$ \cite{thomason}. In time $O(N)$ we can find a forest of $G$'s biconnected components \cite{tarjan}. If we find an $8$-nice decomposition for each biconnected component, join them into a single $8$-nice decomposition by using attachment sets of size $1$ for decompositions inside $G$'s connected component and attachment sets of size $0$ for decompositions in different connected components. Hence, further we assume that $G$ is biconnected.

The $O(N)$ algorithm of \cite{reed} finds a $2$-block tree $\langle T', \mathcal{G}' \rangle$ for $G$ and then for each color $1$ node $G'_t \in \mathcal{G}'$ it finds $(3, 3)$-block tree $\langle T'', \mathcal{G}'' \rangle$ where all components are either planar or M\"{o}bius ladders. To get an $8$-nice decomposition from each $(3, 3)$-block tree, 1) for each color $2$ node contract an edge between it and one of its neighbours in $T''$ and 2) remove all edges which were only created during $\langle T'', \mathcal{G}'' \rangle$ construction (2nd item of $(3, 3)$-block tree definition).

Now we have to draw additional edges in the forest $F$ of obtained $8$-nice decompositions so that to get a single $8$-nice decomposition $\mathcal{T}$ of $G$. Notice that for each pair of adjacent nodes $G'_t, G'_s \in \mathcal{G}'$ where $G'_t$ is color $1$ node and $G'_s = (\{ u, v \}, \varnothing)$ is a color $2$ node, $u, v$ are in $V(G'_t)$ and $\{ u, v \} \in E(G'_t)$. Hence, there is at least one component $G''_{r}$ of $8$-nice decomposition of $G'_t$ where both $u$ and $v$ are present. For each pair of $s$ and $t$ draw an edge between $s$ and $r$ in $F$. Then 1) for each color $2$ node in $F$ (such as $s$) contract an edge between it and one of its neighbors (such as $r$) and 2) remove all edges which were created during $\langle T', \mathcal{G}' \rangle$ construction (2nd item of $2$-block tree definition). This results is a correct $c$-nice decomposition for biconnected $G$.
\end{proof}

\section*{Upper Bound Minimization and Marginal Computation in Approximation Scheme}

Denote:
\begin{equation*}
h(J') \triangleq \min_{\rho(r) \geq 0, \sum_r \rho(r) = 1} g(J', \rho), \qquad g(J', \rho) \triangleq \min_{\{ J^{(r)} \}, \sum_r \rho(r) \hat{J}^{(r)} = J'} \sum_r \rho(r) \log Z(G^{(r)}, 0, J^{(r)})
\end{equation*}
where $h(J')$ is a tight upper bound for $\log Z (G', 0, J')$.

Given a fixed $\rho$, we compute $g(J', \rho)$ using L-BFGS-B optimization \cite{lbfgsb} by back-propagating through $Z(G^{(r)}, 0, J^{(r)})$ and projecting gradients on the constraint linear manifold. On the upper level we also apply L-BFGS-B algorithm to compute $h(J')$, which is possible since \cite{wainwright,globerson}
\begin{equation*}
    \frac{\partial}{\partial \rho(r)} g(J', \rho) = \log Z(G^{(r)}, 0, J^{(r)}_{min}) - (M^{(r)})^\top J^{(r)}_{min}, \,\,\, M^{(r)} \triangleq \frac{\partial}{\partial J^{(r)}_{min}} \log Z(G^{(r)}, 0, J^{(r)}_{min})
\end{equation*}
where $\{ J^{(r)}_{min} \}$ is argmin inside $g(J', \rho)$'s definition and $M^{(r)} = \{ M^{(r)}_e \, | \, e \in E(G^{(r)}) \}$ is a vector of \textit{pairwise marginal expectations}. We reparameterize $\rho(r)$ into $\frac{w(r)}{\sum_{r'} w(r')}$ where $w(r) > 0$.

For $e = \{ v, w \} \in E(G)$ we approximate pairwise marginal probabilities as \cite{wainwright,globerson}
\begin{equation*}
    \mathbb{P}^{alg} (x_v x_w = 1) = \frac12 \cdot \lbrack \sum_r \rho(r) M^{(r)}_e \rbrack + \frac12
\end{equation*}
Let $e_{A}$ be an edge between central vertex $v$ and apex in $G'$. We approximate singleton marginal probability at vertex $v$ as
\begin{equation*}
    \mathbb{P}^{alg} (x_v = 1) = \frac12 \cdot \lbrack \sum_r  \rho(r) M^{(r)}_{e_A} \rbrack + \frac12
\end{equation*}

\section*{Future Work}

We see the following straightforward extensions of the algorithm presented in the manuscript:
\begin{enumerate}
    \item The work \cite{curticapean} extends the polynomial scheme of \cite{straub} for perfect matching counting in a case when $G$ is $H$-free, where $H$ is a \textit{single-crossing} graph, i.e. a minor of an arbitrary graph that can be drawn on a plane with no more than one edge crossing. We claim without proofs that the same applies for a setting considered in this manuscript.

    \item In \cite{straub} authors also present a parallel version of their perfect matching counting scheme in $K_5$-free graphs and show that the problem is in $\text{TC}^1$ parallel complexity class. We claim without proofs that the same applies for inference of $K_5$-free zero-field Ising models.

\end{enumerate}
\end{appendices}

\end{document}